\documentclass[cmp]{svjour}
\usepackage{bbm}
\usepackage{amssymb}
\usepackage{amsmath}
\usepackage{verbatim}
\usepackage{graphicx}
\usepackage{epsfig}

\usepackage{pb-diagram}
\usepackage{lamsarrow}
\usepackage{pb-lams}

\newcommand{\lra}{\longrightarrow}

\newcommand{\uno}{\mathbbm{1}}

\newcommand{\ten}{\otimes}
\newcommand{\N}{\mathbb{N}}
\newcommand{\K}{\mathbb{K}}
\newcommand{\R}{\mathbb{R}}
\newcommand{\C}{\mathbb{C}}

\newcommand{\E}{\mathbb{E}}

\newcommand{\Q}{\mathcal{Q}}
\renewcommand{\L}{\mathcal{L}}

\newcommand{\la}{\lambda}
\newcommand{\al}{\alpha}

\renewcommand{\>}{\rangle}

\DeclareMathOperator{\id}{id} \DeclareMathOperator{\tr}{tr}

\begin{document}

\title{Unbounded violations of bipartite Bell Inequalities via Operator Space theory}

\author{M. Junge \inst{1} \and C. Palazuelos \inst{2}    \and D. P\'{e}rez-Garc\'{\i}a \inst{2} \and I. Villanueva \inst{2} \and M.M. Wolf \inst{3}}

\institute{Department of Mathematics, University of Illinois at
Urbana-Champaign, Illinois 61801-2975, USA \and Departamento de An\'{a}lisis Matem\'{a}tico,
Universidad Complutense de Madrid, 28040, Madrid, Spain \and Niels Bohr Institute, 2100 Copenhagen, Denmark}

\date{\today}

\maketitle

\begin{abstract}
In this work we show that bipartite quantum states with local Hilbert space dimension $n$ can violate a Bell inequality by a factor of order ${\rm
\Omega}\left(\frac{\sqrt{n}}{\log^2n}\right)$ when observables with $n$ possible outcomes are used. A central tool in the analysis is a close relation between this
problem and operator space theory and, in particular, the very
recent noncommutative $L_p$ embedding theory.

As a consequence of this result, we obtain better Hilbert space
dimension witnesses and quantum violations of Bell inequalities  with better resistance to noise.
\end{abstract}

\section{Introduction}

The fact that certain quantum correlations cannot be explained within any local classical theory is one of the most intriguing phenomena
arising from quantum mechanics. It was discovered by Bell \cite{Bell} as
a way of testing the
validity of Einstein-Podolski-Rosen's believe that local hidden
variable models are a possible underlying explanation of physical reality \cite{EPR}. Bell realized that the innocent looking assumptions behind any local hidden variable theory lead to non-trivial restrictions on the strength of correlations. These constraints bear his name and are since called \emph{Bell inequalities} \cite{WernerWolf}. Nowadays, the violation of Bell inequalities in quantum mechanics has
become an indispensable tool in the modern development of
Quantum Information and its applications
cover a variety of areas: quantum cryptography, where it opens the
possibility of getting unconditionally secure quantum key
distribution \cite{Acin1,Acin,Mas2,Mas}; entanglement detection,
where it is the only way of experimentally detecting entanglement
without a priori hypothesis on the behavior of the experiment;
complexity theory, where it enriches the theory of multipartite
interactive proof systems
\cite{Ben-Or,Cleve,Cleve2,Jain,DLTW,KRT,KKMTV}; communication
complexity (see the recent review \cite{Buhrman}); Hilbert space
dimension estimates \cite{Briet,Brunner2,PWJPV,Vertesi,Wehner};
etc.

The violation of Bell inequalities also provides a natural way of \emph{quantifying} the deviation from a local classical description. Unfortunately, computing the maximal violation for a given quantum state or Bell inequality turns out to be a daunting
task  except for very special cases. In \cite{PWJPV} we uncovered
a close connection between {\it tripartite correlation} Bell
inequalities and the mathematical theory of operator spaces,
developed since the 80's as a noncommutative version of the
classical Banach space theory. With these connections at hand, and
with the wide tool-box of operator spaces, we were able to prove the
existence of unbounded violations of tripartite correlation Bell
inequalities.  At the same time this resolved an open problem in pure mathematics related to Grothendieck's famous \emph{fundamental theorem of the metric theory of tensor products}. The relation of Grothendieck's theorem  with correlation Bell inequalities was long ago pointed out by Tsirelson \cite{Tsirelson}.

In the present paper we show how operator spaces are again the appropriate language
to deal with the {\it general bipartite} case, opening in this way
an avenue for the understanding of general bipartite Bell
inequalities. Then, using operator space techniques, we show how to
get violations of ${\rm \Omega}\left(\frac{\sqrt{n}}{\log^2n}\right)$,
using $n$ dimensional Hilbert spaces and $k=n$ outputs. This almost
closes the gap to the ${\rm O}(n)$ (resp. ${\rm O}(k^2)$) upper
bound for such violations given in Proposition \ref{upper-dim}
(resp. in \cite{DKLR}). Again our techniques rely on probabilistic
tools and use the classical random subspaces from Banach space
theory which are now popular in signal processing, see \cite{CDDV}.
The result in this paper implies the existence of
better Hilbert space dimension witnesses and
non-local quantum distributions with a higher resistance to noise --a
desirable property when looking for loophole free Bell tests. Based on the results in \cite{DKLR}, one can also obtains from our result new quantum-classical savings in communication complexity.

\section{Statement of the result}\label{Main Result}

\begin{center}
\begin{figure}
  \includegraphics[width=8cm]{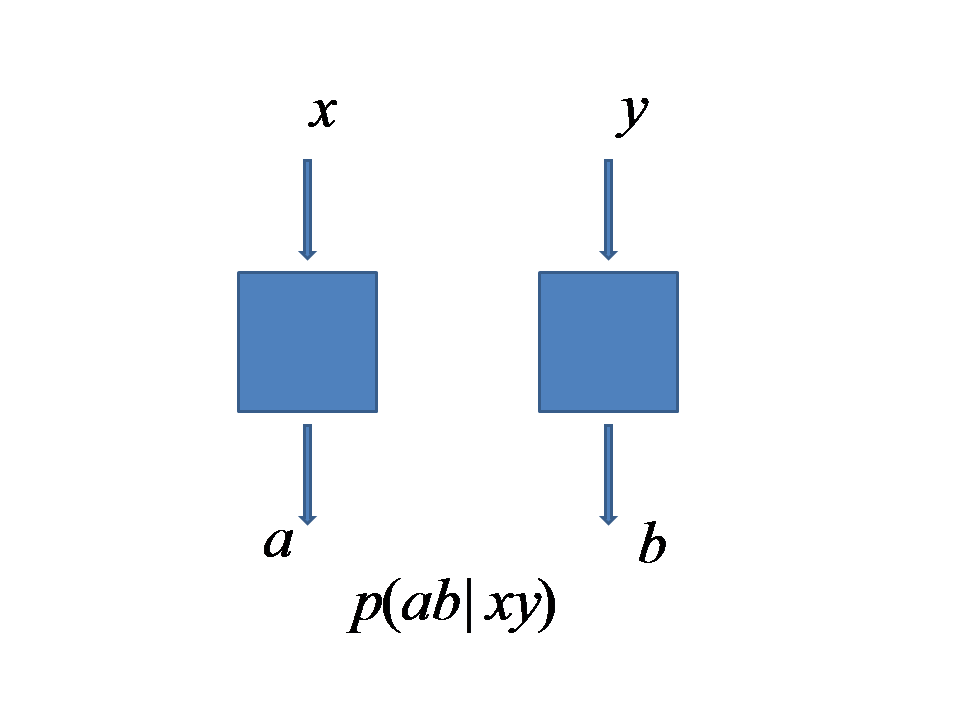}\\
  \caption{$p(ab|xy)$  is the probability distribution of the measurement outcomes $a,b$, if Alice and Bob choose the observables labeled by $x$ and $y$ respectively.}\label{figure1}
\end{figure}
\end{center}
We deal with the following scenario. Alice and Bob represent
spatially separated observers which can choose among different observables
labeled by $x=1,\ldots, N$ in the case of Alice and
$y=1,\ldots, M$ in the case of Bob. The possible measurement outcomes are
labeled by $a=1,\ldots, K$ for Alice and $b=1,\ldots, L$ for Bob. For simplicity we will always assume that $M=N$ and $K=L$. We will refer to the observables $x$ and $y$ as \emph{inputs} and call $a$ and $b$ \emph{outputs}.
The object under study is the probability distribution of $a,b$ given $x,y$, that is, $P(ab|xy)$. Being
a probability distribution, $P(ab|xy)$ verifies
\begin{itemize}
\item $P(ab|xy)\ge 0$\quad (positivity)
\item
$\sum_{ab}P(ab|xy)=1$ for all $x,y$\quad (normalization)
\end{itemize}

In addition, we recall that a probability distribution
$P=p(ab|xy)$ is

\begin{itemize}

\item [a)] \emph{Non-signalling} if
\begin{align*}
\sum_a P(a,b|x,y)&=P(b|y) \text{ is independent of } x, \\
\sum_b P(a,b|x,y)&=P(a|x) \text{ is independent of } y.
\end{align*}
That is, Alice choice of inputs does not affect Bob's marginal
probability distribution and viceversa. This is physically motivated by the principle of \emph{Einstein locality} which implies non-signalling if we assume that Alice and Bob are space-like separated. We denote the set of non-signalling probability
distributions by $\mathcal{C}$.

\item [b)] \emph{Classical} if
\begin{equation}\label{classical}
P(a,b|x,y)=\int_\Omega P_\omega(a|x)Q_\omega(b|y)d\mathbb{P}(\omega)
\end{equation}
for every $x,y,a,b$, where $(\Omega,\mathbb{P})$ is a probability
space, $P_\omega(a|x)\ge 0$ for all $a,x,\omega$, $\sum_a
P_\omega(a|x)=1$ for all $x,\omega$ and the analogue conditions
for $Q_\omega(b|y)$. We denote the set of classical probability
distributions by $\mathcal{L}$.

\item [c)] \emph{Quantum} if there exist two Hilbert spaces $H_1$,
$H_2$ such that
\begin{equation}\label{quantum}
P(a,b|x,y)=tr(E_x^a\otimes F_y^b \rho)
\end{equation}
for every $x,y,a,b$, where  $\rho\in B(H_1\otimes H_2)$ is  a density
operator and $(E_x^a)_{x,a}\subset B(H_1)$, $(F_y^b)_{y,b}\subset B(H_2)$ are two sets
of operators representing POVM
measurements on Alice and Bob systems. That is, $E_x^a\geq 0$ for every  $x,a$, $\sum
_{a}E_x^a=\uno$ for every $x$, $F_y^b\geq 0$ for every $y,b$
and $\sum _{b}F_y^b=\uno$ for every $y$. We denote the set of
quantum probability distributions by $\mathcal{Q}$.

\end{itemize}

It is well known \cite{Tsirelson,DKLR} that
$\mathcal{L}\varsubsetneq\mathcal{Q}\varsubsetneq  \mathcal{C}$
and $\mathcal{C}\subset {\rm Aff}(\mathcal{L})$ with equality if
we restrict to probability distributions. Here,
$${\rm Aff}(\mathcal{L})=\left\{\sum_{i=1}^N \alpha_iP_i:N\in \N,
P_i\in \mathcal{L},\alpha_i\in\R,  \sum_{i=1}^N\alpha_i=1\right\}$$
denotes the affine hull of the space $\mathcal{L}$.

Our aim is to quantify the distance between $\mathcal{Q}$ and
$\mathcal{L}$. For that, we define the `largest Bell violation'
that a given $P\in \mathcal{C}$ may attain  as
$$\nu(P)=\sup\{\langle M,P
\rangle: M \text{  verifies  } |\langle M,P' \rangle|\leq 1 \text{
for every } P'\in \mathcal{L}\},$$ where
$M=\{M_{x,y}^{a,b}\}_{x,y=1,a,b=1}^{N,K}$ is the ``Bell
inequality'' acting on $P$ by duality as $\langle M,P
\rangle=\sum_{x,y,a,b}P(a,b|x,y)M_{x,y}^{a,b}$.

Thus, in order to measure how far is the set $\mathcal{Q}$ from
$\mathcal{L}$, we are interested in computing the maximal possible
Bell violation $$\sup_{P\in\Q}\nu(P).$$

{\bf Notation:} In the whole paper, given a real number $x$ we write
$[x]$ to denote the smallest natural number $p$ such that
$x\leq p$.

Our main result states:

\begin{theorem}\label{Final1}
For every $n\in \N$ and every $2< q<\infty$, there exists a bipartite
quantum probability distribution $P$ with $[n^\frac{q}{2}]^n$
inputs per site, $n+1$ outputs and Hilbert spaces of dimension $n$ each such
that
$$\nu(P)\succeq D(q)n^{\frac{1}{2}-\frac{2}{q}},$$ where $\succeq$
denotes inequality up to a universal constant and $D(q)$ is a
constant depending only on $q$.
\end{theorem}

Actually, by the definition of $\nu$, this result is equivalent to
the following dual formulation

\begin{theorem}\label{Final2}
For every $n\in \N$ and every $2< q<\infty$, we can find a Bell
inequality $M=(M_{x,y}^{a,b})_{x,y,a,b}$, with $x,y=1,\cdots,
[n^\frac{q}{2}]^n$, $a,b=1\cdots, n+1$ such that
$$
\frac{\sup_{P \in \mathcal{Q}}|\langle M,P\rangle|}{\sup_{P \in
\mathcal{L}} |\langle M,P\rangle|}\succeq
D(q)n^{\frac{1}{2}-\frac{2}{q}}.$$ Furthermore, the local Hilbert space
dimension required to get this violation is at most $n$.
\end{theorem}

It follows from the proof of Theorem \ref{Final1} given in Section
\ref{Proof of the main result}, that $D(q)$ can be taken to be
bigger than $\frac{1}{q^2}$. Then, making $q=\log n$ in  Theorem
\ref{Final1} we obtain the following
\begin{corollary}
For every $n\in \N$ there exists a bipartite quantum probability
distribution $P$ with $[2^\frac{\log^2n}{2}]^n$ inputs, $n+1$
outputs and Hilbert spaces each of dimension $n$ such that
$$\nu(P)\succeq \frac{\sqrt{n}}{\log^2n}.$$
\end{corollary}

An analogous consequence holds for Theorem \ref{Final2}.

\section{Upper bounds}

We want to understand how close to optimality Theorem \ref{Final1}
is. In this direction, we present upper bounds to $\nu(P)$ depending
on the number of outputs and the Hilbert space dimension.

First, we have the following result
from \cite{DKLR}, showing a bound for $\nu(P)$ as a function of the number of outputs.

\begin{proposition}\label{Degorre-upper}
Independently of the Hilbert space dimension and the number of
inputs, if $P$ is a quantum probability distribution with $k$
outputs then $$\nu(P)= {\rm O}(k^2).$$
\end{proposition}

If we fix instead the Hilbert space dimension $n$, one can prove the following proposition. A proof is provided in Appendix \ref{upper-dim}.

\begin{proposition}\label{upper-dim}
Independently of the number of inputs and outputs, if $P$ is a bipartite
quantum probability distribution obtained with Hilbert spaces of local
dimension $n$, then $$\nu(P)= {\rm O}(n).$$
\end{proposition}

\section{Prior bipartite unbounded violations}

As pointed out by Tsirelson \cite{Tsirelson}, Grothendieck's Theorem, which he himself called the \emph{fundamental theorem of the metric theory of tensor products},
shows that we can not obtain unbounded violations in the case of
correlation matrices.

The first unbounded violations of Bell inequalities can be traced
back to an application of Raz parallel repetition theorem \cite{Raz}, which
trivially ensures that the parallel repetition of the magic square
game has a violation which grows with $n$ inputs, $n$ outputs
and a Hilbert space of dimension $n$ as $n^x$ for some $x>0$. Similar results hold for any
pseudo-telepathy game \cite{Brassard-review}. Even using the improved
version of Raz theorem given recently in \cite{Holenstein,Rao}, or the concentration theorem given in \cite{Rao}, the
best nowadays available lower bound using this technique
seems to be not much better than $\Omega(n^{10^{-5}})$.

In \cite{KRT}, the authors make a spectacular improvement over
this last quantity. They prove the existence, for each $\nu$, of
unique two provers one round games with $n$ outputs and $2^n/n$
inputs such that the quantum value of the game is larger than
$1-54\nu$ and the classical one smaller than $2/n^\nu$. This involves a
violation of order $\Omega(n^{\frac{1}{54}})$. Their proof strongly relies
on a deep result of Khot and Vishnoi in the context of complexity
theory \cite{KV}.

Therefore, our ${\rm \Omega}\left(\frac{\sqrt{n}}{\log^2n}\right)$ violation
with $n$ outputs and local Hilbert space dimension $n$ can be seen as an
important improvement to the previous results. The prize to pay is
the increase of the number of inputs to ${\rm
O}\left([2^\frac{\log^2n}{2}]^n\right)$.

\section{Resistance to noise}\label{sec:noise}

In the search for a loophole free Bell test, much has been written
about non-locality in the presence of detector inefficiencies (see
for instance
\cite{Brunner,Buhrman,Cabello,Cabello2,Massar,Massar2,Pearle}).
This is modelled in \cite{Massar} by adding an extra output $\perp$
that means ``no detection'' in both Alice and Bob sides. If the
detector efficiency is $\eta$, we then change the ``perfect''
probability distribution $P=P(ab|xy)$ by $\eta^2P+(1-\eta^2)P'$
where $P'=P'(ab|xy)\in\mathcal{L}$ is the local distribution defined
by
$$(1-\eta^2)P'(ab|xy)=\eta(1-\eta)P(a|x)\delta_{b,\perp}+\eta(1-\eta)\delta_{a,\perp}P(b|y)+(1-\eta)^2\delta_{a,\perp}\delta_{b, \perp}.$$

That is, we can interpret the inefficiency of the detector as a
{\it local noise} added to the original probability distribution.
The same happens with other classes of imperfections in the
detectors: for instance if, with certain probability, the detector produces a random output
instead of working properly.

Therefore, in order to have non-local distributions even in the
presence of noise, we fix $P\in\mathcal{C}$ and look at
\begin{equation}\label{Definition of Pi}
\pi(P)=\inf\{\pi: \text{  for all  }P'\in \mathcal{L}
\text{ , } \pi P+(1-\pi)P'\not\in \mathcal{L}\}.
\end{equation}

The following proposition shows that this is ``exactly'' what we are estimating. Specifically,

\begin{proposition}\label{equivalence}
For every $P\in \mathcal{C}$, $$\nu(P)=\frac{2}{\pi(P)}-1.$$
\end{proposition}
By our main result, this proves the existence of quantum probability
distributions with $n$ outputs and Hilbert spaces of dimension $n$
which can withstand any local noise with relative strength
$O\left(1-\frac{\log^2(n)}{\sqrt{n}}\right)$ (see next section). It
is interesting to note that, by Proposition \ref{upper-dim},
${\rm O}\left(1-\frac{1}{n}\right)$ is an upper bound for the maximal possible resistance
to noise. However, if one restricts exclusively to the noise coming
from inefficient detectors, one can obtain {\it exponential}
resistance \cite{Massar}. It is time for the  proof of Proposition
\ref{equivalence}.

\begin{proof} Let $P\in \mathcal{C}$.  We refer to \cite{DKLR} for the  fact
 that
 \begin{equation}\label{relation}
 \nu(P)= \inf\{\sum_{i=1}^I |\alpha_i|: P=\sum_{i=1}^I
 \alpha_iP_i,P_i\in \mathcal{L},\alpha_i\in\R,
 \sum_{i=1}^I\alpha_i=1 \}.
 \end{equation}
Let $\la=\pi(P)$. By definition we have  $\la  P+(1-\la)P'=P''$ is
again in $\mathcal{L}$. This gives
$P=\frac{1}{\la}P''-(\frac{1}{\la}-1)P'$ and therefore, by Equation
(\ref{relation}), $\nu(P)\le \frac{2}{\la}-1$.

For the converse we use again Equation (\ref{relation}) and start
with the decomposition $P=\sum_i\alpha_i P_i$ such that
$\sum_i|\alpha_i|=\nu(P)$. Dividing in positive and negative terms
we get $P=\sum_i\alpha_i^+P_i-\sum_i\alpha_i^-P_i$, where $\sum_i
\alpha_i^+-\sum_i\alpha_i^-=1$ and $\sum_i
\alpha_i^++\sum_i\alpha_i^-=\nu(P)$. Let us denote by
$r=\sum_i\alpha_i^+$ the positive part. Hence we have
$2r=\nu(P)+1$ and  therefore
 \begin{equation}\label{eq1}
 \frac{1}{r}P+(1-\frac{1}{r})P'=P''\;, P',P''\in\mathcal{L}.
 \end{equation}
Indeed, $P'=\frac{\sum_i\alpha_i^-P_i}{\sum_j\alpha_j^-}$ and
$P''=\frac{\sum_i\alpha_i^+P_i}{\sum_j\alpha_j^+}$.  Equation
(\ref{eq1}) gives $\la\ge\frac{1}{r}$. Since $2r=\nu(P)+1$ we
obtain that $\nu(P)\ge\frac{2}{\la}-1$, which concludes the proof.
\end{proof}

\section{Incomplete probability distributions}

We present here {\em incomplete probability distributions}. We need
them for the statement and proof of Theorem \ref{mainresult1}. Our
main result, Theorem \ref{Final1}, will follow as a corollary. We
also use incomplete probability distributions to formalize the
treatment given to noise in the previous section.

We are interested in computing $\pi(P)$ when we consider local
probability distributions $P'$ with $k+1$ outputs in Equation
(\ref{Definition of Pi}), where $k$ is the number of outputs of
$P$. To this end  we embed $P$ into the space of probability
distribution of $k+1$ outputs just by adding the corresponding $0$'s
and denoting the new distribution by $\tilde{P}$. We denote $\mathcal L_{k}$ to the
local distributions with $k$ outputs (the other parameters are
fixed). By Proposition \ref{equivalence}, we  compute

$$\nu(\tilde{P})=\sup_M\frac{|\<M,\tilde{P}\>|}{\sup_{P'\in\mathcal{L}_{k+1}}|\<M,P'\>|}.$$

Of course, restricting with $M$'s which vanish on the index given by the extra output $\perp$ will give
a lower bound for $\nu(\tilde{P})$. That is, we have

$$\nu(\tilde{P})\geq \sup_M\frac{|\<M,P\>|}{\sup_{P'\in\mathcal{L}_k}|\<M,P'\>|},$$
where $P'$ is now of the form
 \begin{equation}\label{LHVM}
P'(a,b|x,y)=\int_\Omega
P_\omega(a|x)Q_\omega(b|y)d\mathbb{P}(\omega).
\end{equation}
$(\Omega,\mathbb{P})$ is a  probability space and for every
$\lambda$, $x$ (resp. $y$) $(P(a|x, \omega))_x^a$ (resp.
$(Q(b|y,\omega))_y^b$) is a sequence of positive numbers such that
$\sum_a P(a|x, \omega))_x^a\leq 1$ (resp. is $\sum_a Q(a|x,
\omega))_x^a\leq 1$). We will say that such a $P'$ is an
\emph{incomplete classical probability distribution}.

In this section we deal with this kind of {\it incomplete}
probability distributions and prove a generalization of our main result,
Theorem \ref{Final1}, to this setting. This will formalize the claim
stated in Section \ref{sec:noise} concerning the existence of quantum
probability distributions with $n$ outputs and Hilbert spaces of
dimension $n$ which can withstand any local noise with extra
outputs and relative strength $O(1-\frac{\log^2(n)}{\sqrt{n}})$.

The rest of the paper is essentially devoted to prove the above mentioned
generalization, from which Theorem \ref{Final1} can be
deduced.

We say that $P$ is an \emph{incomplete quantum probability
distribution} if there exist two Hilbert spaces $H_1$,
$H_2$ such that
\begin{equation}\label{quantum}
P(a,b|x,y)=tr(E_x^a\otimes F_y^b \rho)
\end{equation}
for every $x,y,a,b$, where  $\rho\in B(H_1\otimes H_2)$ is  a density
operator and $(E_x^a)_{x,a}\subset B(H_1)$, $(F_y^b)_{y,b}\subset B(H_2)$ are two sets
of operators representing {\em incomplete} POVM
measurements on Alice and Bob systems. That is, $E_x^a\geq 0$ for every  $x,a$, $\sum
_{a} E_x^a\leq \uno$ for every $x$, $F_y^b\geq 0$ for every $y,b$
and $\sum _{b}E_y^b \leq \uno$ for every $y$.

We denote the set of incomplete quantum distributions by
$\mathcal{Q}^{in}$ and the set of incomplete classical distributions
by $\mathcal{L}^{in}$.

With these definitions at hand, we can introduce
\begin{definition}\label{largestviolation}
Given a linear functional (Bell inequality)
$M=(M_{x,y}^{a,b})_{x,y=1,a,b=1}^{N,K}$,  we define the
\emph{Classical bound} of $M$ as the number
$$B_C(M)=\sup\{|\langle M,P\rangle|: P\in \mathcal{L}^{in}\}$$ and the
\emph{Quantum bound} of $M$ as $$B_Q(M)=\sup\{|\langle M,P\rangle|: P \in
\mathcal{Q}^{in}\}.$$ We define the \emph{largest quantum violation of
$M$} as the positive number
\begin{equation}
LV(M)=\frac{B_Q(M)}{B_C(M)}.
\end{equation}

\end{definition}

\begin{remark}
It is easy to see that $B_C(M)=0$ implies $B_Q(M)=0$ for every $M$. We will rule out these cases because they lack interest.
\end{remark}

The generalization of Theorem \ref{Final1} to this context is the following one.

\begin{theorem}\label{mainresult1}
For every $n\in \N$ and every $2< q<\infty$, we can find a linear
functional $M=(M_{x,y}^{a,b})_{x,y,a,b}$, $x,y=1,\cdots,
[n^\frac{q}{2}]^n$, $a,b=1\cdots, n$ such that $$LV(M)\succeq
D(q)n^{\frac{1}{2}-\frac{2}{q}}.$$ The local Hilbert space dimension required to get this violation is at most $n$.
\end{theorem}

Our next two results follow straightforwardly.

\begin{corollary}
For every $n\in \N$ we can find a linear functional
$M=(M_{x,y}^{a,b})_{x,y,a,b}$, $x,y=1,\cdots, [2^\frac{\log^2n}{2}]^n$, $a,b=1\cdots, n$ such that
$$LV(M)\succeq \frac{\sqrt{n}}{\log^2n}.$$ The local Hilbert space
dimension needed to get this violation is at most $n$.
\end{corollary}

\begin{corollary}
For all $n$, there exists a probability distribution $P$ with
$[2^\frac{\log^2n}{2}]^n$ inputs, $n+1$ outputs and Hilbert space
dimension $n$ which can withstand any local noise with extra
outputs and relative strength $O(1-\frac{\log^2n}{\sqrt{n}})$.
\end{corollary}

Finally, the next lemma allows us to prove Theorem \ref{Final2} (and, thus, Theorem
\ref{Final1}) from Theorem \ref{mainresult1}.

\begin{lemma}\label{uncomplted-completed}
Suppose we have a linear functional $(M_{x,y}^{a,b})_{x,y,a,b}$,
$x,y=1,\cdots ,N$, $a,b=1,\cdots  ,K$ such that $LM(V)=C$. Then,
there exists another linear functional
$(\hat{M}_{x,y}^{a,b})_{abxy}$, $x,y=1,\cdots ,N$, $a,b=1,\cdots
,K+1$ such that
$$\frac{\sup_{P\in\Q} |\langle \hat{M},P\rangle|}{\sup_{P\in\L}|\langle \hat{M},P\rangle|}=C.$$
\end{lemma}

\begin{proof}
It is enough to define $\hat{M}$ as the extension of $M$ for which
$\hat{M}_{x,y}^{K+1,b}=0$, $\hat{M}_{x,y}^{a,K+1}=0$.
\end{proof}

\section{Bounds for the Hilbert space dimension}

The interest in testing the Hilbert
space dimension started with a crucial observation made in
\cite{Acin}. In that paper, the authors observe that the standard security proofs for the
BB84 protocol \cite{Shor,Kraus} assume a given dimension in the
Hilbert space and  they can fail if this assumption is
dropped. In \cite{Brunner2}, motivated by that, the authors
define the concept of ``dimension witness'' and show some
examples in low dimensions. Since then, several contributions to the field have
appeared with different approaches: Bell inequalities
\cite{Briet,Vertesi}, quantum random access codes \cite{Wehner} or
quantum evolutions \cite{Wolf}.

We define $\Q_d$ to be the distributions in $\Q$ with the extra
restriction that the Hilbert spaces $H_1,H_2$ appearing in the
definition are $d$-dimensional. With this notation, a dimension
witness for dimension $d$ is simply a ``Bell inequality'' $M_{d,n}$ such
that $|\<M_{d,n},P_d\>|\le C_d$ for all $P_d\in \Q_d$,  and for such
that  there exists $P\in \Q_n$ with $|\<M_{d,n},P\>|> C_d$. In the case of
binary outcomes, Bri\"et, Buhrman and Toner \cite{Briet} and Vertesi
and Pal \cite{Vertesi} have shown how to get dimension estimates for
any dimension. However, in their case
$$\sup_{M_{d,n}}\frac{\sup_{P_n\in\Q_n}|\<M_{d,n},P_n\>|}{\sup_{P_d\in
\Q_d}|\<M_{d,n},P_d\>|}\in [1,K_G].$$
This means that the resolution of the considered witnesses is bounded by Grothendieck's constant $K_G$ and indeed could vanish with increasing dimension.




It would be therefore desirable to get
\begin{equation}\label{dimension}
\sup_{M_{d,n}}\frac{\sup_{P_n\in\Q_n}|\<M_{d,n},P_n\>|}{\sup_{P_d\in
\Q_d}|\<M_{d,n},P_d\>|}\underset{n\ge d\rightarrow\infty}{\longrightarrow}\infty.
\end{equation}
For two outcomes this was shown to be possible in the {\it
tripartite} case \cite{PWJPV}. Our main
Theorem, together with Theorem \ref{upper-dim} implies that
\begin{theorem}
For any $d,n$ we can define dimension estimates $M_{d,n}$
verifying \begin{equation*}
\sup_{M_{d,n}}\frac{\sup_{P_n\in\Q_n}|\<M_{d,n},P_n\>|}{\sup_{P_d\in
\Q_d}|\<M_{d,n},P_d\>|}= {\rm
\Omega}\left(\frac{\sqrt{n}}{\log^2(n)d}\right).
\end{equation*}
\end{theorem}

\section{Mathematical tools and Connections}\label{Mathematical tools and Connections}

In this section we will introduce the basic notions about operator spaces which we will need along this work. We do recommend \cite{EffrosRuan} and \cite{Pisierbook} for a much more complete reference.

\

The theory of operator spaces was born with the work of Effros and Ruan in the 80's, see for instance \cite{EffrosRuan,Pisierbook}. They characterized, in an abstract sense, the structure of the closed subspaces of $B(H)$, the space of bounded linear operators on a Hilbert space.

\

Formally, an operator space is a complex vector space $E$ and a sequence of
norms $\|\cdot\|_n$ in the space of $E$-valued matrices
$M_n(E)=M_n\otimes E$, which verify the following two properties
\begin{enumerate}

\item For every $n,m\in\N$, $x\in M_m(E)$, $a\in M_{nm}$ and $b\in M_{mn}$ we have that
$$\|axb\|_n\le \|a\|\|x\|_m\|b\|$$
\\

\item For every $n,m\in \N$, $x\in M_n(E)$, $y\in M_m(E)$, we have that $$\left\|\left(%
\begin{array}{cc}
 x & 0 \\
 0 & y \\
\end{array}%
\right)\right\|_{n+m}= \max\{\|x\|_n,\|y\|_m\}.$$
\end{enumerate}

Any $C^*$-algebra $\mathcal{A}$ has a natural operator space structure induced by its natural embedding $j:\mathcal{A}\hookrightarrow B(H)$. Indeed, it is enough to consider the sequence of norms on $M_n\otimes \mathcal{A}$ defined by the embedding $id\otimes j:M_n\otimes \mathcal{A}\hookrightarrow M_n\otimes B(H)=B(\ell_2^n\otimes_2H)$. In
particular, $\ell_{\infty}^k$ has a natural operator space structure.
To compute it we isometrically embed $\ell_{\infty}^k$ into the diagonal of $M_k$ and then, given $x=\sum_i A_i \otimes
e_i\in M_n(\ell_{\infty}^k)=M_n\otimes \ell_{\infty}^k$, we have
\begin{equation}\label{eq:op-space-infty}
\|x\|_n=\left\|\sum_{i} A_i \otimes
|i\rangle\langle i|\right\|_{M_{nk}}=\max_i\|A_i\|_{M_n}.
\end{equation}

\

In order to attain a better understanding of  the differences
between Banach space category and operator space category, we need
to look not only at the spaces, but also at the morphisms, that is,
the operations which preserves the structure. We
will have to consider now the so called {\em completely bounded
maps}. They are linear maps $u:E\lra F$ between operator spaces such
that all the dilations $u_n=\uno_n \otimes u :M_n\otimes E=
M_n(E)\lra M_n\otimes F = M_n(F)$ are bounded. The cb-norm of $u$ is
then defined as $\|u\|_{cb}=\sup_n\|u_n\|$. We will call $CB(E,F)$
the resulting normed space. It has a natural operator space
structure induced by $M_n(CB(E,F))=CB(E, M_n(F))$. We can
analogously define the notion of a complete isomorphism/isometry (see
\cite{EffrosRuan,Pisierbook}).

\smallskip

The so called minimal tensor product of two operator spaces
$E\subset B(H)$ and $F\subset B(K)$ is defined as the operator
space $E\otimes_{\min} F$ with the structure inherited from the
induced embedding $F\otimes E\subset B(H\otimes K).$ In
particular, $M_n(E)=M_n\otimes_{\min} E$ for every operator space
$E$. The tensor norm $\min$ in the category of operator spaces
will play the role of the so called $\epsilon$ norm in the
classical theory of tensor norms in Banach spaces \cite{Def}. In particular
$\min$ is injective, in the sense that if $E\subset X$ and
$F\subset Y$ completely isomorphic/isometric, then
$E\otimes_{\min} F\subset X\otimes_{\min} Y$ completely
isomorphic/isometric. The analogue in the operator space category
of the $\pi$ tensor norm is the projective tensor norm, defined as
$$\|u\|_{M_n(E\otimes^{\wedge}
F)}=\inf\{\|\alpha\|_{M_{n,lm}}\|x\|_{M_l(E)}\|y\|_{M_m(F)}\|\beta\|_{M_{lm,n}}:u=\alpha(x\otimes
y)\beta\},$$ where $u=\alpha(x\otimes y)\beta$ means the matrix
product  $$u=\sum_{rsijpq}\alpha_{r,ip} \beta_{jq,s}
 |r\rangle\langle s| \otimes x_{ij}\otimes y_{pq}  \in M_n\otimes
 E\otimes F.$$
Both tensor norms, ${\wedge}$ and $\min$, are associative and
commutative and they share the duality relations which verify
$\pi$ and $\epsilon$ in the context of Banach spaces. In
particular, for finite dimensional operator spaces we have the
natural completely isometric identifications
\begin{equation}\label{eq:dual-op-space}
(E\otimes^{\wedge} F)^*=CB^2(E,F; \C)=CB(E,F^*)=E^*\otimes_{\min}
F^*,
\end{equation}
where, given an operator space $E$, we define its dual operator space $E^*$ via the identification $M_n(E^*)=CB(E,M_n)$.

\smallskip

Given a Banach space $X$, we can consider in it different operator
space structures or, equivalently, different isometric embeddings of
$X$ into $B(H)$ which lead to different families of matrix norms.
For example we may embed an $n$-dimensional Hilbert space as column
 \[ C_n =\{ \sum_k \alpha_k |k\rangle \langle 0| :\al_k\in
 \mathbb{C} \} \quad \mbox{or row space}\quad
  R_n =\{ \sum_k \alpha_k |0\rangle \langle k| :\al_k\in
 \mathbb{C} \} \]
Let us note that
 \[  \|\sum_i A_i\otimes e_i\|_{M_m\otimes_{\min} R_n}=\|\sum_i
 A_iA_i^{\dagger}\|^{\frac{1}{2}},
 \|\sum_i A_i\otimes e_i \|_{M_m\otimes_{\min} C_n}=\|\sum_i
 A_i^{\dagger}A_i\|^{\frac{1}{2}}.\]
Using matrices of the form $A_i=|i\rangle \langle 0|$ we deduce the
well-known fact that this yields different matrix norms (see
\cite{Pisierbook} from more details).

\smallskip

The natural operator space structure on $\ell_1^n$ is the one obtained by the duality
$(\ell_\infty^n)^*=\ell_1^n$ and it can be seen that for every
operator space $X$ the space $\ell_\infty\otimes_{min}X$ (resp.
$\ell_1\otimes^{\wedge}X$) coincides, as a Banach space, with
$\ell_\infty\otimes_{\epsilon}X=\ell_\infty(X)$ (resp.
$\ell_1\otimes_\pi X=\ell_1(X))$. Furthermore, for every operator
space $X$, the natural operator space structure defined on $\oplus_1^n X$ (see
\cite{Pisierbook}) allows us to identify completely isometrically
this operator space with $\ell_1^n \otimes^{\wedge} X$ via the
natural identification. This operator space is denoted by $\ell_1^n
(X)$. Analogous reasonings hold for the operator space
$\ell_\infty^n (X)$. Actually, by the comments above, it follows
that $$(\ell^n_\infty(X))^*=\ell^n_1(X^*)\text{  (completely
isometrically)  }$$ for every finite dimensional operator space $X$.

\

The operator space $L_p$-embedding theory has been  developed in
the last years. Some of the most important results of classical
Banach space theory, as well as probability theory and harmonic
analysis have found analogous versions in the noncommutative case
\cite{JPX,JungeParcet1,JungeParcet2}.

\

As we did in a previous work \cite{PWJPV}, we will reduce the
problem of separating the Classical from the Quantum probability
distributions to the problem of separating the epsilon and the min
norm on the tensor products of certain operator spaces (see Section
\ref{connections}). The noncommutative $L_p$-embedding theory will
allow us to find ``good'' subspaces where we can compute the
above  mentioned tensor norms.

\

Our new tool in this paper are spaces constructed as sums and
intersection in interpolation theory. These spaces already play an
important role for embedding problems in operator space theory \cite{JungeParcet3,JungeParcet2}. For fixed
$t>0$ and $m\in \mathbb{N}$,  we  consider the operator space
$K(t;\ell_\infty^m, R^m+C^m,\ell_1^m)$ defined by  the matrix norm
 \begin{align*}
 & \|x\|_{M_n(K(t;\ell_\infty^m, R^m+C^m,\ell_1^m))}\\
 &=\inf_{x=x_1+x_2+x_3}\{\|x_1\|_{M_n(\ell_\infty^m)}+\sqrt{t}\|x_2\|_{M_n(R^m+C^m)}+
 t \|x_3\|_{M_n(\ell_1^m)}\}
 \end{align*}
As in classical interpolation theory \cite[Lemmas 3.1, 3.5]{JungeParcet3}, it is easy to determine the  dual space:
\begin{equation}\label{dual}
K(t;\ell_\infty^m,
R^m+C^m,\ell_1^m)^*\sim J(t^{-1}; \ell_1^m, R^m\cap
C^m,\ell_\infty^m),
\end{equation}
where $J(t^{-1}; \ell_1^m, R^m\cap
C^m,\ell_\infty^m)$ denotes the operator space given by
  \begin{align*}
 &\|a\|_{M_n(J(t^{-1}; \ell_1^m, R^m\cap
 C^m,\ell_\infty^m))}\\
 &=  \max \{\|a\|_{M_n(\ell_1^m)}, t^{-\frac{1}{2}}\|a\|_{M_n(R^m\cap
 C^m)},t^{-1}\|a\|_{M_n(\ell_\infty^m)}\}.
 \end{align*}
Here, $\sim$ denotes a complete isomorphism up to a universal
constant (in this case $16$). The following result will be crucial
in our work:

\begin{theorem}[\cite{JungeParcet3}, Theorem 3.6]\label{embedding}
Let $(\Omega,\mu)$ be a measure space such that $\mu(\Omega)=n$. Then, the application $$j:L_1(\Omega)+ L_2^r(\Omega)+ L_2^c(\Omega)+L_\infty(\Omega)\hookrightarrow L_1(\Omega^n;\ell_\infty^n),$$ defined by  $$j(f)(\omega_1,\cdots,\omega_n)=\frac{1}{n^n}\sum_{k=1}^nf(\omega_k)e_k$$ is a complete embedding (with absolute constants).

\end{theorem}

This result is stated in a much
more general context in \cite{JungeParcet3}. In Appendix \ref{sec:appendix2} we boil the
rather heavy notation down to the result used here.

\subsection{Connection to the ``$min$ vs $\varepsilon$ problem''}\label{connections}

In this section we will connect the Classical (resp. Quantum) bounds of a given linear functional $M$ (see Definition  \ref{largestviolation}) to two natural tensor norms in the framework of classical Banach spaces and operator spaces.

\

We associate with a four dimensional matrix with coefficients
$M=(M_{x,y}^{a,b})_{x,y=1,a,b=1}^{N,K}$ the corresponding tensor
$$\sum_{x,y=1,a,b=1}^{N,K}M_{x,y}^{a,b}(e_x\otimes e_a)\otimes (e_y\otimes e_b)$$ considered as an element of
$\ell_1^N(\ell_\infty^K)\otimes \ell_1^N(\ell_\infty^K)$. Our next
result deals with the Classical bound:
\begin{proposition}\label{epsilon}
Given $M=(M_{x,y}^{a,b})_{x,y=1,a,b=1}^{N,K}$, we have the following equivalence
$$B_C(M)\leq \|M\|_{\ell_1^N(\ell_\infty^K)\otimes_\epsilon
\ell_1^N(\ell_\infty^K)}\leq 4B_C(M).$$
\end{proposition}

\begin{proof}
By duality, it follows that $$\|M\|_{\ell_1^N(\ell_\infty^K)\otimes_\epsilon
\ell_1^N(\ell_\infty^K)}=\sup\left\{\sum_{a,x,b,y}M_{x,y}^{a,b}T_{x,y}^{a,b}:  \, T_{x,y}^{a,b}\in
B_{\ell_\infty^N(\ell_1^K)\otimes_{\pi}\ell_\infty^N(\ell_1^K)}\right\}.$$
Since  $B_{\ell_\infty^N(\ell_1^K)\otimes_{\pi}\ell_\infty^N(\ell_1^K)}$
is the convex hull of the set $\{x\otimes y: x\in B_{\ell_\infty^N(\ell_1^K)}, y\in B_{\ell_\infty^N(\ell_1^K)} \}$, we have that

$$\|M\|_{\ell_1^N(\ell_\infty^K)\otimes_\epsilon
\ell_1^N(\ell_\infty^K)}=\sup\left\{\sum_{a,x,b,y}M_{x,y}^{a,b}\int_\Omega
 P_\omega(x,a)Q_\omega(y,b)d\mathbb{P}(\omega)\right\},$$ where the $\sup$ is taken over all

\begin{enumerate}
\item[a)] $(\Omega, \mathbb{P})$ probability space,

\item[b)] $\sum_{a=1,\cdots ,K}|P_\omega(x,a)|\leq 1$ for every $x=1,\cdots ,N$ and every $\omega$,

\item[c)]$\sum_{b=1,\cdots ,K}|Q_\omega(y,b)|\leq 1$ for every $y=1,\cdots ,N$ and every $\omega$.
\end{enumerate}

Using this, the first inequality follows. For the second one it
is enough to consider the positive and negative part of each
$P_\omega(x,a)$ and $Q_\omega(y,b)$.

\end{proof}

Next we deal with the Quantum bound:

\begin{theorem}\label{minimal}

Given $M=(M_{x,y}^{a,b})_{x,y=1,a,b=1}^{N,K}$, we have the following equivalence

$$B_Q(M)\leq \|M\|_{\ell_1^N(\ell_\infty^K)\otimes_{min}
\ell_1^N(\ell_\infty^K)}\leq 16B_Q(M).$$

\end{theorem}

Before we prove the result, let us note that
\begin{equation}\label{min-norm}
\|M\|_{\ell_1^N(\ell_\infty^K)\otimes_{min}
\ell_1^N(\ell_\infty^K)}=\sup \{\|(u\otimes
v)(M)\|_{B(H)\otimes_{min} B(H)}\},
\end{equation}
where the $\sup$ is taken over all the operators
$u:\ell_1^N(\ell_\infty^K)\rightarrow B(H)$ which verify
$\|u\|_{cb}\leq 1$ (and the same for $v$). We will use the
following Lemma:

\begin{lemma}\label{positive}
Let $(T_n)_n\subset B(H)$ be a sequence of positive operators.
Then $ \|\sum_nT_n\|_{B(H)}=\|\sum_n T_n\otimes e_n
\|_{B(H)\otimes_{min}\ell_1}$, where we are considering the natural operator space structure on $\ell_1$.
\end{lemma}

\begin{proof}
It can be seen \cite[Prop. 8. 9]{Pisierbook} that, for every sequence $(a_n)_n$
in $B(H)$, we have
 \[ \|\sum_n a_n\otimes e_n \|_{B(H)\otimes_{min}\ell_1}=\inf
 \{\|\sum_nb_nb_n^*\|^{\frac{1}{2}}
 \|\sum_nc_n^*c_n\|^{\frac{1}{2}}\},\]
where the $\inf$ is taken over all possible decompositions
$a_n=b_nc_n$. Now, if we take $b_n=c_n=(T_n)^{\frac{1}{2}}$, we
obtain
 \[ \|\sum_n T_n\otimes e_n \|_{B(H)\otimes_{min}\ell_1}\leq
  \|\sum_nT_n\|_{B(H)}.\]
On the other hand, it is known \cite[Prop. 8. 9]{Pisierbook} that
the norm of  $\|\sum_n T_n\otimes e_n
\|_{B(H)\otimes_{min}\ell_1}$ is equal to
 \[ \sup\left\{\left\|\sum_n T_n\otimes
 U_n\right\|_{B(H)\otimes_{min}B(H)}:U_n\in B(H),
 U_nU_n^*=U_n^*U_n=\uno\right\}.\]
Then, taking $U_n=\uno$ for every $n$, we get
$$\left\|\sum_nT_n\right\|_{B(H)} \leq\|\sum_n T_n\otimes
e_n \|_{B(H)\otimes_{min}\ell_1}.$$ Alternatively, this
follows from the fact that the functional $\Sigma:\ell_1\to
\mathbb{C}$, $\Sigma((t_n))=\sum_n t_n$ is a complete contraction.
\end{proof}

The following remark will make the proof of Theorem \ref{minimal} easier to read

\begin{remark}\label{min-BQ}
Note that, using the isometric identification
\begin{equation}\label{identification}
CB(\ell_1^N(\ell_\infty^K),B(H))=\ell_\infty^N(\ell_1^K)\otimes_{min}B(H)=\ell_\infty^N(\ell_1^K\otimes_{min}B(H)),
\end{equation}we can deduce from the previous lemma that $B_Q(M)$ is exactly the same as $\|M\|_{\ell_1^K(\ell_\infty^N)\otimes_{min}
\ell_1^K(\ell_\infty^N)}$, when we consider operators $u$ and $v$ which map the canonical basis of $\ell_1^K(\ell_\infty^N)$ to positive elements of $B(H)$ in Equation (\ref{min-norm}).

Indeed, this is immediate from the two following facts. First, given a complete contraction $u:\ell_1^N(\ell_\infty^K)\rightarrow B(H)$ such that $u(e_x\otimes e_a)=G_x^a\in B(H)^+$ for every $x=1,\cdots, N;a=1,\cdots, K$, we will have that, for every $x$, $$\|\sum_{a=1}^K G_x^a\|=\|\sum_{a=1}^K G_x^a\otimes e_a\|_{B(H)\otimes_{min}\ell_1^k}\leq \sup_x\|\sum_{a=1}^K G_x^a\otimes e_a\|_{B(H)\otimes_{min}\ell_1^k}\leq 1,$$ where we have used Lemma \ref{positive} in the first equality and Equation (\ref{identification}) in the last inequality.  $\sum_{a=1}^K G_x^a$ being a positive element for every $x$, the above estimation implies that we have a sequence of operators $(G_x^a)_{x=1,\cdots,N}^{a=1,\cdots ,K}\subset B(H)^+$ such that $\sum_{a=1}^K G_x^a\leq \uno$ for every $x=1,\cdots, N$. On the other hand, given a sequence $(E_x^a)_{x=1,\cdots,N}^{a=1,\cdots ,K}\subset B(H)^+$ such that $\sum_{a=1}^K E_x^a\leq \uno$ for every $x=1,\cdots, N$, we can consider the operator $u:\ell_1^N(\ell_\infty^K)\rightarrow B(H)$ defined by $u(e_x \otimes e_a)=E_x^a\in B(H)^+$. Using again Lemma \ref{positive} and Equation (\ref{identification}) we can see that $$\|u\|_{cb}=\sup_x\|\sum_{a=1}^K E_x^a\otimes e_a\|_{B(H)\otimes_{min}\ell_1^k}=\sup_x\|\sum_{a=1}^K E_x^a\|\leq 1.$$
\end{remark}

We can prove now Theorem \ref{minimal}.

\begin{proof}[of Theorem \ref{minimal}] The first inequality is just the previous remark.
For the proof of the second inequality we consider a complete
contraction $u:\ell_1^N(\ell_\infty^K)\rightarrow B(H)$. We consider
now the $u_x:\ell_{\infty}^K\to B(H)$. We recall that according to
Wittstock's factorization theorem \cite[Theorem 8.5]{Paulsen} every complete contraction $u$ defined
on a $C^*$-algebra $\mathcal{A}$ with values in $B(H)$ can be
decomposed as  $u(x)=V\pi(x)W$,  with $V$ and $W$ contractions, and
$\pi$ a $^*$-representation. Thus, we have
$u=u^1_1-u^2_2-i(u^3-u^4)$, where
 \begin{align*}
  u^1_1(x)&=1/4(V+W^*)\pi(x)(V^*+W)\, ,\,
 u^2_2(x)=1/4(V-W^*)\pi(x)(V^*-W)\, , \,\\
 u^3_3(x)&=1/4(V-iW^*)\pi(x)(V^*+iW)\, , \,
 u^4(x)=1/4(V+iW^*)\pi(x)(V^*-iW)
 \end{align*}
Note that, for every $i=1,...,4$,  $u^i$ is a completely positive
contraction.

We apply this observation to every component and we decompose
$u_x=u_x^1-u_x^2-i(u^3_x-u^4_x)$ as a linear combination of
completely positive maps. Then, for every $x$ and $a$ we see that
$(u_x^i(e_a))_{a}=(E_{x,i}^a)_a$ is an incomplete POVM (see also
Equation (\ref{identification})). This leads to the constant
$16=4\times 4$ in the assertion.
\end{proof}

The following corollary follows now from the previous two theorems:

\begin{corollary}\label{mainconnection}
Given $M=(M_{x,y}^{a,b})_{x,y=1,a,b=1}^{N,K}$, we have that $$LV(M)\simeq \frac{\|M\|_{\ell_1^N(\ell_\infty^K)\otimes_{min}
\ell_1^N(\ell_\infty^K)}}{\|M\|_{\ell_1^N(\ell_\infty^K)\otimes_\epsilon
\ell_1^N(\ell_\infty^K)}},$$ where $\simeq$ denotes equality up to universal constants.
\end{corollary}

\section{Proof of the main result}\label{Proof of the main result}

We introduce some notation that will be useful in the proof.

\begin{remark}\label{level1}
Given an operator space $X$, we construct an associated operator
space $X^n$ as follows: Let $I$ be the collection of all complete
contractions $v:X\lra M_n$. Then, we can define a new operator space
structure on the Banach space $X$ considering the application
$$j:X\longrightarrow \ell_\infty(I,M_n)$$ defined by
$$j(x)=((v(x))_{v\in I}.$$

It is easy to see that
$$M_n(X^n)=M_n(X).$$
For our purpose it is interesting to note
that
$$\|a\|_{X^n\otimes_{min}Y^n}=\sup_{\|v:X\lra M_n\|_{cb}\leq 1, \|w:Y\lra
M_n\|_{cb}\leq 1} \|(v\otimes w)(x)\|_{M_n\otimes_{min} M_n}.$$
\end{remark}

Then, the result that we will prove is

\begin{theorem}\label{MainMath}
Given $2< q<\infty$ and $n\in \mathbb{N}$, take $m$ such that $n^{\frac{q}{2}}\leq m \leq 2n^{\frac{q}{2}}$ (for instance $m=[n^{\frac{q}{2}}]$) and denote $X=\ell_1^{m^n}(\ell_\infty^n)$. Then, we can find an element $x\in X\otimes X$ of rank $n$ such that $\|x\|_{X\otimes_\epsilon X}\leq D(q)$ and  $\|x\|_{X^n\otimes_{min} X^n}\geq n^{\frac{1}{2}-\frac{2}{q}}.$
\end{theorem}

Theorem \ref{mainresult1} follows now from Theorem \ref{MainMath} and Corollary \ref{mainconnection}.

For reference purposes, we state next Chevet's inequality, which will be used often in the following. For a proof see \cite{LedouxTalagrand}.

\begin{theorem}[Chevet's inequality]\label{chevet}
There exists a universal constant $b$ such that for every Banach spaces $E, F$ and every sequence $(g_{s,t})_{s,t}$ of independent normalized gaussian random variables, we have $$\|\sum_{s,t} g_{s,t}
x_s\otimes y_t\|_{E\otimes_{\varepsilon}F} \le bw_2((x_s)_s;
E)\|\sum_{t} g_ty_t\|_F+bw_2((y_t)_t; F)\|\sum_s g_s x_s\|_E,$$where, given a sequence $(x_s)_s$ in a Banach space $X$, we use the notation $w_2((x_s)_s;
X)$ for $$w_2((x_s)_s;
X)=\|\sum_s x_s\otimes e_s\|_{X\otimes_\epsilon \ell_2}=\sup \left\{ \left(\sum_s|x^*(x_s)|^2\right)^\frac{1}{2}:x^*\in X^*, \|x^*\|\leq 1\right\}.$$

We can take $b=1$ if the spaces are real, whereas $b=4$ if they are complex.
\end{theorem}

We will need the following three technical lemmas.

\begin{lemma}\label{min}
Let $1<q<\infty$, $n\leq m$ and $(g_{ij})_{i,j=1}^{n,m}$ a family of
normalized gaussian random variables. Consider
$X_t^q=t^{-\frac{1}{q}}K(t;\ell_\infty^m, R^m+C^m,\ell_1^m)$ for
$t=\frac{n}{m}$. Then,
$$\mathbb{E}\left\|\sum_{i,j=1}^{n,m}g_{ij}e_j\otimes
e_i:X_t^q\lra R_n\cap C_n\right\|_{cb}\leq K
m^{1-\frac{1}{q}}n^{\frac{1}{q}}C(m,n),$$ where $K$ is a universal
constant and $C(m,n)=1+\frac{\sqrt{\log(m)}}{n}$.
\end{lemma}

\begin{proof}
It follows from equation (\ref{dual}) that
\begin{align}
 & \|a:X_t^q\lra R_n\|_{cb} \nonumber \\
 & =\|a\|_{R_n(X_t^*)}\leq
 t^{\frac{1}{q}}\max\{\|a\|_{R_n(\ell_1^m)},
 t^{\frac{-1}{2}}\|a\|_{R_n(R_m\cap
 C_m)},t^{-1}\|a\|_{R_n(\ell_\infty^m)}\}. \label{bound}
\end{align}
where $c$ is a universal constant. We have to estimate the three
terms appearing in this maximum. Recall our use of $\preceq$ for
inequalities valid up to a universal constant. For the first term,
we use the little Grothendieck theorem \cite[Page 183]{MariusHabi}, which says that there exists a
constant $k$ such that for every operator $a:\ell_\infty^m\lra
\ell_2^n$ we have $\|a:\ell_\infty^m\lra R_n\|_{cb}\leq k
\|a\|_{op}$ (and the same for $C_m$). Then, we invoke Chevet's
inequality and  obtain
 \[ \mathbb{E}\left\|\sum_{i,j=1}^{n,m} g_{ij}e_j\otimes
 e_i\right\|_{R_n(\ell_1^m)}\preceq
 \mathbb{E}\left\|\sum_{i,j=1}^{n,m} g_{j,i}e_i\otimes
 e_j\right\|_{\ell_1^m\otimes_{\epsilon} \ell_2^n}\preceq
 (\sqrt{m}\sqrt{n} + m)\leq K_1 m.\]
For the second term, it is easy to see  that
 \[  \mathbb{E}\left \|\sum_{i,j=1}^{n,m} g_{ij}e_j\otimes e_i\right\|_{R_n(R_m\cap
 C_m)}=\mathbb{E}\left\|\sum_{i,j=1}^{n,m} g_{ij}e_j\otimes
 e_i\right\|_{\ell_2^n\otimes_2 \ell_2^m}\preceq \sqrt{nm}.\]
Finally, we will use Chevet's inequality again to estimate the
last expression,
 \[  \mathbb{E}\left\|\sum_{i,j=1}^{n,m} g_{ij}e_j\otimes
 e_i\right\|_{R_n(\ell_\infty^m)}=\mathbb{E}\left\|\sum_{i,j=1}^{n,m} g_{ij}e_j\otimes
 e_i\right\|_{\ell_2^n\otimes_\epsilon \ell_\infty^m}\preceq\sqrt{n} +
 \sqrt{\log m},\]
where we have used that
$\mathbb{E}\left\|\sum_{i=1}^mg_ie_i\right\|_{\ell_{\infty}^m}\preceq
\sqrt{\log m}$ \cite[Page 15]{Tomczak}. Let us
insert the precise value of $t=\frac{n}{m}$. Then we obtain
 \[ \mathbb{E}\|a:X_t^q\lra R_n\|_{cb}\leq c
 t^{\frac{1}{q}}\mathbb{E}[\max\{\|a\|_{R_n(\ell_1^m)},
 t^{\frac{-1}{2}}\|a\|_{R_n(R_m\cap
 C_m)},t^{-1}\|a\|_{R_n(\ell_\infty^m)}\}]$$$$= c
 \mathbb{E}[(\frac{n}{m})^{\frac{1}{q}}\max\{\|a\|_{R_n(\ell_1^m)},
 (\frac{n}{m})^{\frac{-1}{2}}\|a\|_{R_n(R_m\cap
 C_m)},(\frac{n}{m})^{-1}\|a\|_{R_n(\ell_\infty^m)}\}]$$$$\leq
 c\mathbb{E}[(\frac{n}{m})^{\frac{1}{q}}(\|a\|_{R_n(\ell_1^m)}+
 (\frac{n}{m})^{\frac{-1}{2}}\|a\|_{R_n(R_m\cap
 C_m)}+(\frac{n}{m})^{-1}\|a\|_{R_n(\ell_\infty^m)})]$$$$\leq K'
 (\frac{n}{m})^{\frac{1}{q}}(m+ m+(\sqrt{n} + \sqrt{\log m})(\frac{m}{n}))\leq K m^{1-\frac{1}{q}}
 n^{\frac{1}{q}}(1+\frac{\sqrt{\log(m)}}{n}),\]
where $K$ is a universal constant. Replacing $R_n$ by $C_n$ we
find the same estimates. By the definition of the intersection
$R_n\cap C_n$ we obtain the result.
\end{proof}

\begin{lemma}\label{gaussian}
There exits $\delta\in (0,1/2)$ with the following property:
Given natural numbers $n\leq m$ and a family of normalized
gaussian random variables $(g_{ij})_{i,j=1}^{n,m}$, we consider
$G=\frac{1}{\sqrt{m}}\sum_{i,j=1}^{n,m}g_{ij}e_i\otimes e_j$ as an operator from $\ell_2^n$ to $\ell_2^m$. Then, ``with high probability'', there exists an operator $v:H_n\lra \ell_2^n$ such that $v \frac{1}{m}G^*G|_{H_n}=\uno_{H_n}$ and $\|v\|\leq 2$, where we denote $H_n=\ell_2^{[\delta n]}.$
\end{lemma}

\begin{proof}

Chevet's inequality tells us that

$$\mathbb{E}[\|G\|_{op}]\leq a
\frac{1}{\sqrt{m}}(\sqrt{n}+\sqrt{m})\leq C$$for some universal
constant $C$. On the other hand, it is known \cite[Page 80]{MarcusPisier} that
 \[ \mathbb{E}[\|G\|_2]\geq c\sqrt{n}\]
for a universal constant $c>\frac{1}{\sqrt{2}}$, where
$\|\cdot\|_2$ denotes the Hilbert-Schmidt norm. Thus, we can
choose constants (independent of $n$) $0<c<C$ such that, with high
probability, $G$ verifies $\|G\|_{op}\leq C$, and $\|G\|_2\geq c
\sqrt{n}$. We define $\delta=\frac{c^2}{2C^2}$. We recall the
notation $s_j(G)$ for the $j^{th}$ singular value of $G$ and
observe that
 \[ c^2n\leq \|G\|_2^2=\sum_{j=1}^n s_j(G)^2\leq s_1(G)^2([\delta
 n]-1)+s_{[\delta n]}(G)^2n\leq \frac{c^2 n}{2}+s_{[\delta
 n]}(G)^2n.\]
Therefore, we find  $0<\frac{c^2}{2}\leq s_{[\delta n]}(G)^2$.
We may take $c^2>\frac{1}{2}$, so we have $\frac{1}{2}\leq
s_{[\delta n]}(G)$. By the definition of the singular values of
$G$, the above estimation says that we can invert the operator
$\frac{1}{m}G^*G:\ell_2^n\lra \ell_2^n$ on a ``large'' subspace of
dimension $k_n=[\delta n]$. Thus, if we denote
$H_n=\ell_2^{[\delta n]}$, we know that there exists an operator
$v_n:H_n\lra \ell_2^n$ such that $v_n
\frac{1}{m}G^*G|_{H_n}=\uno_{H_n}$ and $\|v_n\|\leq 2$.
\end{proof}

Before we prove the next lemma, let us observe the following remark.

\begin{remark}\label{identity}
By the definition of the $K$-spaces and the standard interpolation
equality $[\ell_{\infty}^m,\ell_1^m]_{1/q}=\ell_q^m$, it is clear
that for every $t>0$, $m\in\N$ and $1< q<\infty$ the map
$\id=\id\circ \id:\ell_q^m\hookrightarrow K(t;\ell_\infty^m,
\ell_1^m)\hookrightarrow K(t;\ell_\infty^m, R^m+C^m,\ell_1^m)$ is
a composition of two contractions (see for instance \cite{Sh} for
the first one), hence $\id$ is itself a contraction.
\end{remark}

\begin{lemma}\label{epsilon}
Given $2< q<\infty$, there exists a constant $c(q)>0$ such that for every
$n\leq m^{\frac{2}{q}}$ and every family of
normalized gaussian random variables $(g_{ij})_{i,j=1}^{n,m}$, we have
$$\mathbb{E}\left\|m^{-\frac{1}{q}}\sum_{i,j=1}^{n,m}g_{ij}e_i\otimes e_j:\ell_2^n\lra X_t^q\right\|\leq c(q)$$(for every $t>0$).
\end{lemma}

\begin{proof} Applying Chevet's inequality again for $X_t^q$, we get
 \[ \mathbb{E}\left\|\sum_{i,j=1}^{n,m}g_{ij}e_i\otimes e_j\right\|_{\ell_2^n\otimes\epsilon X_t^q}\preceq
 \mathbb{E}\left\|\sum_{j=1}^mg_{j}e_j\right\|_{X_t^q}+\sqrt{n}w_2((e_j)_{j=1}^m;
 X_t^q).\]
Hence, it suffices to show that
 \begin{align*}
  \|id:\ell_2^m\lra X_t^q\| &\leq  A(q) \quad \mbox{and} \\
  \mathbb{E}\left\|\sum_{j=1}^mg_{j}e_j\right\|_{X_t^q}&\leq B(q)m^{\frac{1}{q}}.
  \end{align*}
Both estimations follow easily using Remark \ref{identity}.
Indeed, the upper estimate follows, with $A(q)=1$, from
 \[ \|id:\ell_2^m\lra \ell_q^m\|\leq 1. \]
In the same way, the next estimate follows from
 \begin{align*}
  \mathbb{E}\left\|\sum_{j=1}^mg_{j}e_j\right\|_{X_t^q}& \leq
 \mathbb{E}\left\|\sum_{j=1}^mg_{j}e_j\right\|_{\ell_q^m}\leq B(q)
 m^{\frac{1}{q}}.
 \end{align*}
\end{proof}

\begin{remark}\label{constq1}
It is well known that $B(q)\leq C\sqrt{q}$, where $C$ is a
universal constant independent of $q$. Thus, we have that $c(q)\leq C'
\sqrt{q}$.
\end{remark}

Using this, we can separate the epsilon and the min norm on a suitable subspace of $\ell_1^{m^n}(\ell_\infty^n).$

\begin{theorem}\label{construction}
Given $2< q<\infty$ and $n\in \mathbb{N}$, if we take
$n^{\frac{q}{2}}\leq m \leq 2n^{\frac{q}{2}}$, there exists a
matrix $a\in X_t^q\otimes X_t^q$ of rank $n$ such that
$\|a\|_\epsilon\leq D(q)$ and  $\|a\|_{min}\geq
n^{\frac{1}{2}-\frac{2}{q}},$ where we define $t=\frac{n}{m}$ and
$X_t^q=t^{-\frac{1}{q}}K(t;\ell_\infty^m, r^m+c^m,\ell_1^m).$

\end{theorem}

\begin{proof}
Given $n\in\mathbb{N}$ and $2< q<\infty$, taking
$n^{\frac{q}{2}}\leq m \leq 2n^{\frac{q}{2}}$, we define $t$ and
$X_t^q$ as in the statement of the theorem. Since $t$ is
considered fixed we may simplify the notation and write
$X_q=X_t^q$. Thanks to the three previous Lemmas, we know that
there exists a matrix $G=(g_{ij}(w))_{i,j=1}^{n,m}$ such that
 \begin{enumerate}
 \item[1)] $\|G^*:X_q\lra R_n\cap C_n\|_{cb}\leq
 C(q)m^{\frac{1}{q'}}n^{\frac{1}{q}}.$
 \item[2)] There exist $\delta$, $v_n$ and $H_n$ as in Lemma \ref{gaussian}.
 \item[3)] $\|G:\ell_2^n\lra X_q\|\leq c(q)m^{\frac{1}{q}}.$
\end{enumerate}
Observe that, due to the choice of $m$, the function $C(n,m)=C(q)$
(in Lemma \ref{min}) only depends on $q$. Consider an arbitrary
matrix $a$ in $H_n\otimes H_n$. Then, we have
 \[ \|m^{-\frac{1}{q}}G\otimes m^{-\frac{1}{q}}G(a)\|_{X_q\otimes_\epsilon X_q}\leq
 c(q)^2\|a\|_{H_n\otimes_\epsilon H_n}. \]
On the other hand, we have
 \begin{align*}
   \|a\|_{\ell_2^n\otimes_2 \ell_2^n}&   =  \|v(\frac{1}{m}G^*G)\otimes
  v(\frac{1}{m}G^*G)(a)\|_{\ell_2^n\otimes_2\ell_2^n} \\
  &=
  \|v(\frac{1}{m}G^*G)\otimes
   v(\frac{1}{m}G^*G)(a)\|_{R_n\cap C_n\otimes_{min} R_n\cap
   C_n}\\
 &\leq\|vm^{-\frac{1}{q'}}G^*\|_{cb}^2\|(m^{-\frac{1}{q}}G\otimes
 m^{-\frac{1}{q}}G)(a)\|_{X_q^n\otimes_{min} X_q^n}.
 \end{align*}
In the special case where $a$ represents the identity on $H_n$ we
obtain
 \begin{align*}
 \sqrt{n}&\leq
 \delta^{-\frac{1}{2}}\sqrt{k_n}=\delta^{-\frac{1}{2}}\|a\|_{\ell_2^n\otimes_2\ell_2^n}\\
 & \leq
4\delta^{-\frac{1}{2}}C(q)^2n^{\frac{2}{q}}\|(m^{-\frac{1}{q}}G\otimes
 m^{-\frac{1}{q}}G)(a)\|_{X_q^n\otimes_{min} X_q^n}.
 \end{align*}
This leads to the two competing estimates
\begin{align}
  \|m^{-\frac{1}{q}}G\otimes m^{-\frac{1}{q}}G(a)\|_{X_q\otimes_\epsilon X_q} &\leq    c(q)^2 \,
  \label{epsilon}\\
    \|(m^{-\frac{1}{q}}G\otimes
  m^{-\frac{1}{q}}G)(a)\|_{X_q^n\otimes_{min} X_q^n}& \geq
 \frac{D}{C(q)^2}n^{\frac{1}{2}-\frac{2}{q}}. \label{minnorm}
\end{align}
Combing  (\ref{epsilon}) and (\ref{minnorm}) yields the result.
\end{proof}

\begin{remark}\label{level2} According to Remark \ref{level1},
we have actually proved that
\[ \|a\|_{X_q^n\otimes_{min} X_q^n}\geq
\frac{D}{C(q)^2}n^{\frac{1}{2}-\frac{2}{q}}\, .\]
\end{remark}

\begin{remark}\label{constq2}
The constant in the previous theorem can be taken $D(q)\leq C c(q)^2C(q)^2$, where $C$ is a universal constant which does not depend on $q$. Furthermore, we have seen in Remark \ref{constq1} that $c(q)^2\preceq q$. It can be checked that $C(q)\leq 1+\frac{\sqrt{q\log(n)}}{n}.$
\end{remark}

We can prove now Theorem \ref{MainMath}.

\begin{proof}
By Theorem \ref{embedding}, for every measure space $(\Omega,
\mu)$ such that $\mu(\Omega)=k<\infty$, we have that $L_1(\Omega)+
L_2^R(\Omega)+L_2^C(\Omega)+L_\infty(\Omega)$ completely embeds
into $L_1(\Omega^k;\ell_\infty^k)$. Furthermore, the complete
embedding $$j:L_1(\Omega)+
L_2^r(\Omega)+L_2^c(\Omega)+L_\infty(\Omega)\hookrightarrow
L_1(\Omega^k;\ell_\infty^k)$$ can be specifically written. Indeed,
consider the measure space $(\Omega, \mu)$, where
$\Omega=\{1,\cdots,m\}$ and $\mu(i)=t=\frac{n}{m}$ for every
$i=1,\cdots, m$. Then, $\mu(\Omega)=mt=n$. But it is easy to see
that for this measure space, the operator space $L_1(\Omega)+
L_2^r(\Omega)+L_2^c(\Omega)+L_\infty(\Omega)$ is exactly the
operator space $K(t;\ell_\infty^m, R^m+C^m,\ell_1^m)$. Thus, we have
a completely isomorphic embedding of $K(t;\ell_\infty^m,
R^m+C^m,\ell_1^m)$ into $L_1(\Omega^n,\ell_\infty^n)$. Note that the
difference between  $L_1(\Omega^n,\ell_\infty^n)$ and
$X=\ell_1^{m^n}(\ell_\infty^n)$ is just the normalization in the
$L_1$-norm and hence the spaces are completely isometrically
isomorphic. 
Thus, it will be enough to consider the completely isomorphic
embedding $\tilde{j}=r\circ t^{-\frac{1}{q}}j$ from
$t^{-\frac{1}{q}}K_t$ into $X$ and to take the element
$x=(\tilde{j}\otimes \tilde{j})(a)\in X\otimes X$, where $a$ is
the same element as in Theorem \ref{construction}. We invoke
Remark \ref{level2} and the fact that the formal identity map
$id:X\to X^n$ is completely contractive. This yields the difference
for the $\min$ and $\varepsilon$ norm claimed in the assertion.
\end{proof}

\begin{remark}
It follows from Remark \ref{constq2} that we can take $D(q)\leq
q^2$ (actually, this estimate is not tight). Then, for a fixed
dimension $n$, just taking $q=\log(n)$, we obtain
 \begin{equation}
 \frac{\|x\|_{X^n\otimes_{min} X^n}}{\|x\|_{X\otimes_\epsilon X}}
 \geq \frac{\sqrt{n}}{\log(n)^2} \label{log}
 \end{equation}
with $X=\ell_1^{[2^\frac{\log^2n}{2}]^n}(\ell_\infty^n)$.
\end{remark}

\begin{remark} We have the following interesting alternatives: either

 \begin{enumerate}
 \item[a)] for every subspace $F\subset L_1(\ell_\infty)$
  \[ \ell_2\ten_{\varepsilon} F \, =\,  R+C \ten_{\min} F \, \]

or

\smallskip

 \item[b)] there exists a subspace $F\subset L_1(\ell_\infty)$
 such that
   \[ \ell_2\ten_{\varepsilon} F \neq  R+C \ten_{\min} F \, \]
 \end{enumerate}
In case a), it follows easily from John's theorem \cite{Pisier3} that for every rank $n$ tensor $a\in
F_1\otimes  F_2$ that
 \[ \|a\|_{\min} \le C \sqrt{n} \|a\|_{\varepsilon} \, \]
This means our  estimate \eqref{log} for a rank $n$ tensor is
optimal up to the logarithmic factor. However, in case b) there
are violations of Bell's inequality involving POVM's only for
Alice or Bob, but not both. To wrap this up we could formulate it
as follows. \emph{Either there are assymmetric Bell violations
which are of simpler nature than everything discovered so far, or
our estimates are best possible}. It would certainly be
interesting to know which of these alternatives holds true.
\end{remark}



\section*{Acknowledgments}

The authors are grateful to the organizers of the \emph{Operator Structures in Quantum Information Workshop}, held in Toronto during July 6-10, 2009; where part of this work was developed. M. Junge is partially supported by the NSF grant DMS-0901457. C. Palazuelos, D. Perez-Garcia and I. Villanueva are partially supported by Spanish grants I-MATH, MTM2008-01366 and CCG08-UCM/ESP-4394. M.M.
Wolf acknowledges support by QUANTOP and the Danish Natural Science
Research Council(FNU).

\appendix

\section{Some proofs}\label{sec:appendix2}

\subsection{Proof of Proposition \ref{upper-dim}}\label{sec:appendix1}

The result is based on the fact that the norm of the identity
$$id:M_n\otimes_\epsilon M_n\rightarrow M_n\otimes_{min} M_n$$ is $\leq n$ (actually it is exactly $n$). Indeed, using that $$d_{cb}(R_n, min(\ell^n_2))=d_{cb}(C_n, min(\ell^n_2))=\sqrt{n},$$ it is easy to see that $d_{cb}(M_n,min(M_n))=n$. The result follows now trivially from the fact $min(M_n)\otimes_{min} M_n=M_n\otimes_\epsilon M_n$.

Let us take then a Bell inequality $M=\{M_{x,y}^{a,b}\}_{x,y,a,b}$ and a quantum probability
distribution $P$. By the previous estimation, we have

$$|\langle M,P\rangle |\le \left\|\sum_{a,b,x,y}M_{x,y}^{a,b}E_a^x\otimes F_b^y\right\|_{M_n\otimes_{min}M_n}\leq$$ $$n \left\|\sum_{a,b,x,y}M_{x,y}^{a,b}E_a^x\otimes F_b^y\right\|_{M_n\otimes_{\varepsilon}M_n}.$$ Now, this is exactly the same as \begin{equation}\label{ap1}\sup\left\{\left|\sum_{a,b,x,y}M_{x,y}^{a,b}\tr(E_a^x\rho_1) \tr(F_b^y\rho_2)\right|:\rho_1,\rho_2\in B_{S_1^n}\right\}.\end{equation}

But it is well known that every $\rho\in B_{S_1^n}$ can be written as $\rho=\rho_1^1+i\rho_1^2$ with $\rho_1^i$ self adjoint elements in $B_{S_1^n}$ for $i=1,2$. Then, $$(\ref{ap1})\le 4 \sup\{|\sum_{a,b,x,y}M_{x,y}^{a,b}\tr(E_a^x\rho_1) \tr(F_b^y\rho_2)|:\rho_1,\rho_2 \in B_{S_1^n} \text{  and self adjoint  }\}.$$

But $\rho_1$ can be written as  $\rho_1=\sum_{j=1}^n\delta_j |f_j\>\<f_j|$ with $(|f_j\>)_j$ an orthonormal basis of $\ell_2^n$, and $\sum_{j=1}^n|\delta_j|\leq 1$ (and the same for $\rho_2$). Then, for every pair of selfadjoint $\rho_1,\rho_2$ we have

 \begin{align*}
 &\left|\sum_{a,b,x,y}M_{x,y}^{a,b} \tr(E_a^x\rho_1)
 \tr(F_b^y\rho_2)\right |\\
 & \leq \sup\left\{\left|\sum_{a,b,x,y}M_{x,y}^{a,b}\<u|E_a^x|u\>\<v| F_b^y|v\>\right|: |u\>,|v\>\in
 S_{\ell_2^n}\right\},
 \end{align*}
which is bounded above by $\sup_{P'\in
\mathcal{L}}|\langle M,P'\rangle|.$ Therefore, we have
$$|\langle M,P\rangle| \preceq n \sup_{P'\in \mathcal{L}}|\langle
M,P'\rangle|.$$

\subsection{Explanation of [\cite{JungeParcet3}, Theorem 3.6]}\label{sec:appendix2}

Suppose we have a probability space $(\Omega,\mu)$ and $k\in\N$.
We may consider the particular case of (Theorem 3.6,
\cite{JungeParcet3}) in which $\mathcal{A}=M_k\otimes_{min}
L_\infty(\Omega^n)$, $\mathcal{M}=M_k\otimes_{min}
L_\infty(\Omega)$, $\mathcal{N}=M_k$, the conditional expectation
$\mathcal{E}_\mathcal{N}:\mathcal{M}\rightarrow \mathcal{N}$ is
defined by $\mathcal{E}_\mathcal{N}=\uno\otimes \int \cdot\, d\mu$
and $\mathcal{K}=\mathbb{C}$. The algebras
$(\mathcal{M})_{k\geq1}$'s form a system of independent symmetric
system of copies of $\mathcal{M}$ over $\mathcal{N}$ (see
(\cite{JungeParcet3}, Example 1), which is a stronger condition
than the one appearing in (Theorem 3.6, \cite{JungeParcet3}). We
start with the easy case
 \[ L_1(\mathcal{A},\ell_\infty^n)=L_1(M_k\otimes
 L_\infty(\Omega^n),\ell_\infty^n)=S_1^k(L_1(\Omega^n),\ell_\infty^n).
 \]
Let us turn to the more complicated $\K$ space
 \begin{align*}
  &\mathcal{K}_{1,\infty}^n(\mathcal{M},
  \mathcal{E}_\mathcal{N})\, = \,
    nL_1(\mathcal{M})+ L_1^s(\mathcal{M},
  \mathcal{E}_\mathcal{N}) +
  \sqrt{n}L_1^r(\mathcal{M},
  \mathcal{E}_\mathcal{N})+\sqrt{n}L_1^c(\mathcal{M},
 \mathcal{E}_\mathcal{N}) \, .
 \end{align*}
Here we refer to definition before (\cite{JungeParcet3}, Lemma
3.5)
 \[ \|x\|_{L_1^s(\mathcal{M},
 \mathcal{E}_\mathcal{N})}=\inf_{x=ayb}\|a\|_{L_2(M_k)}\|y\|_{M_k\otimes_{min}
 L_\infty(\Omega)}\|b\|_{L_2(M_k)}=\|x\|_{S_1^k(L_\infty(\Omega))}
 \, . \]
Hence $L_1^s(\mathcal{M})=S_1^k(L_1(\Omega))$ as predicted.  For
the column term we have
 \begin{align*}
   \|x\|_{L_1^c(\mathcal{M}, \mathcal{E}_\mathcal{N})}
  &=\inf_{x=ayb}\|a\|_{L_2(\mathcal{M})}\|y\|_{M_k\otimes
  L_\infty(\Omega)}\|b\|_{L_2(M_k)}\\
  &= \inf_{x=ab}\|a\|_{L_2(\mathcal{M})}
   \|b\|_{L_2(M_k)} \, .
   \end{align*}
Given such a factorization $x=ab$ we see that
 \begin{align*}
  \|(\int_{\Omega} |x|^2 d\mu)^{1/2}\|
  &= \| b^*(\int_{\Omega} a^*a d\mu) b\|_{1/2}^2 \\
  &\le
  \|b\|_{L_2(M_k)} \|\int_{\Omega} tr(a^*a)d\mu\|_1^{1/2}
  \, = \, \|b\|_2 \|a\|_2 \,
  \end{align*}
This shows $\|x\|_{S_1^k(L_2^r(\Omega))}\le \inf \|a\|\|b\|$.
Conversely, for $x\in S_1^k(L_2^r(\Omega))=R_k\ten_h
L_2^r(\Omega)\ten_h C_k$ we deduce from the definition of the
Haagerup tensor product that we can find a factorization $x=ba$
such that $b\in R_k\ten_h L_2^r(\Omega)\ten_h R_k$ and $a\in
L_2(M_k)$. Note however, that
 \[ \|b\|_{R_k\ten_h L_2^r(\Omega)\ten_h R_k}
 =  \|b\|_{L_2(\Omega,S_2^k)}
 = \|b\|_{L_2(\mathcal{M})} \, . \]
Thus we have in fact
 \[ \|x\|_{S_1^k(L_2^r(\Omega))} = \|x\|_{L_1^c(\mathcal{M},
 \mathcal{E}_\mathcal{N})} \]
Interchanging rows and columns yields the missing estimate.
Theorem \ref{embedding} follows now easily. Suppose we have a
measure space $(\Omega,\mu)$ such that $\mu(\Omega)=n$. Then, we
consider $(\Omega,\hat{\mu})=(\Omega,\frac{\mu}{n})$ and, thus, $$
i: n
L_1(\Omega,\hat{\mu})+\sqrt{n}L_2^r(\Omega,\hat{\mu})+\sqrt{n}L_2^c(\Omega,\hat{\mu})+L_\infty(\Omega,
\hat{\mu})\hookrightarrow L_1(\Omega^n,\otimes^n\hat{\mu};
\ell_\infty^n),$$ is a completely embedding. But it is obvious
that
$$nL_1(\Omega,\hat{\mu})=L_1(\Omega,\mu),
\sqrt{n}L_2^r(\Omega,\hat{\mu})=L_2^r(\Omega,\mu),
\sqrt{n}L_2^c(\Omega,\hat{\mu})=L_2^c(\Omega,\mu),$$$$L_\infty(\Omega,\hat{\mu})=L_\infty(\Omega,\mu)
\text{  and  } n^nL_1(\Omega^n,\otimes^n\hat{\mu};
\ell_\infty^n)=L_1(\Omega^n,\otimes^n\mu; \ell_\infty^n).$$

Therefore, $$j=\frac{i}{n^n}:L_1(\Omega,\mu)+ L_2^r(\Omega,\mu)+L_2^c(\Omega,\mu)+L_\infty(\Omega,\mu)\hookrightarrow L_1(\Omega^n,\otimes^n\mu; \ell_\infty^n)$$ is a complete embedding (with absolute constants).


\begin{thebibliography}{99}


\bibitem{Acin1}  A. Acin, N. Brunner, N. Gisin, S.
Massar, S. Pironio, V. Scarani, \emph{Device-independent security
of quantum cryptography against collective attacks}, Phys. Rev.
Lett. 98, 230501 (2007).

\bibitem{Acin0} A. Acin, N. Brunner, N. Gisin, S. Massar, S. Pironio, V. Scarani, \emph{Device-independent security of quantum cryptography against collective attacks}, Phys. Rev. Lett., \textbf{98}, 230501 (2007).

\bibitem{Acin} A. Acin, L. Masanes, N. Gisin, \emph{From Bell's Theorem to Secure Quantum Key
Distribution}, Phys. Rev. Lett. 97, 120405 (2006).



\bibitem{Bell} J.S. Bell, \emph{On the Einstein-Poldolsky-Rosen paradox}, Physics, \textbf{1}, 195 (1964).

\bibitem{Sh} C. Bennett, R. Sharpley, \emph{Interpolation of operators}, Academic Press, 1988.

\bibitem{Ben-Or} M. Ben-Or, A. Hassidim, H. Pilpel, \emph{Quantum Multi Prover Interactive Proofs with Communicating
Provers}, Proceedings of 49th Annual IEEE Symposium on Foundations
of Computer Science (FOCS 2008), arXiv:0806.3982.

\bibitem{Brassard-review}  G. Brassard, A. Broadbent, A. Tapp, \emph{Quantum
Pseudo-Telepathy} Foundations of Physics, Volume 35, Issue 11, Nov
2005, Pages 1877 - 1907.

\bibitem{Briet} J. Briët, H. Buhrman, B. Toner, \emph{A generalized Grothendieck inequality and entanglement in
XOR games}, arXiv:0901.2009.

\bibitem{Brunner} N. Brunner, N. Gisin, V. Scarani,
C. Simon, \emph{Detection loophole in asymmetric Bell
experiments}, Phys. Rev. Lett., 98, 220403 (2007).

\bibitem{Brunner2} N. Brunner, S. Pironio, A. Acin, N.
Gisin, A. A. Methot, V. Scarani, \emph{Testing the Hilbert space
dimension}, Phys. Rev. Lett. 100, 210503 (2008).


\bibitem{Buhrman} H. Buhrman, R. Cleve, S. Massar, R. de Wolf,
\emph{Non-locality and Communication Complexity},  to appear in
Reviews of Modern Physics.


\bibitem{Cabello} A. Cabello, J.-A. Larsson, \emph{Minimum detection efficiency for a loophole-free atom-photon Bell
experiment}, Phys. Rev. Lett. 98 (2007) 220402.

\bibitem{Cabello2} A. Cabello, D. Rodriguez, I. Villanueva, \emph{Necessary and sufficient detection efficiency for the Mermin
inequalities}, Rev. Lett. 101 (2008) 120402.

\bibitem{Cleve} R. Cleve, P. Høyer, B. Toner, and J. Watrous,
\emph{Consequences and Limits of Nonlocal Strategies}, Proceedings
of the 19th IEEE Annual Conference on Computational Complexity
(CCC 2004), pp. 236- 249 (2004).

\bibitem{Cleve2} R. Cleve, D. Gavinsly, R. Jain,\emph{Entanglement-Resistant Two-Prover Interactive Proof Systems
and Non-Adaptive Private Information Retrieval Systems},
quant-ph/07071729, 2007.

\bibitem{CDDV} A. Cohen, W. Dahmen and R. DeVore, \emph{Compressed sensing and best k-term
approximation}, JAMS, 22, No. 1,(2009), 211-231.

\bibitem{Def} A. Defant and K. Floret,
\emph{Tensor Norms and Operator Ideals}, North-Holland, Amsterdam, 1993.
\bibitem{DKLR} J. Degorre, M. Kaplan, S. Laplante, J. Roland,
\emph{The communication complexity of non-signaling distributions}, arXiv:0804.4859.

\bibitem{DLTW}  A. C. Doherty, Y-C. Liang, B. Toner, S. Wehner, \emph{The quantum moment problem and
bounds on entangled multi-prover games}, Proceedings of IEEE Conference on Computational Complexity 2008, pages 199-210.

\bibitem{EffrosRuan} E. G. Effros and Z.-J. Ruan, {\it Operator
spaces}, London Math. Soc. Monographs New Series, Clarendon Press,
Oxford, 2000.

\bibitem{EPR} A. Einstein, B. Podolsky, N. Rosen,
\emph{Can Quantum-Mechanical Description of Physical Reality Be
Considered Complete?}, Phys. Rev., {\textbf 47}, 777 (1935).




\bibitem{Grothendieck} A. Grothendieck, \emph{R\'esum\'e de la th\'eorie m\'etrique des produits tensoriels
topologiques} (French), Bol. Soc. Mat. São Paulo, \textbf{8},
1-79, (1953).

\bibitem{Holenstein} T. Holenstein, \emph{Parallel repetition: simplifications and the no-signaling
case}, Proceedings of the thirty-ninth annual ACM symposium on
Theory of computing (STOC) 2007.



\bibitem{Jain}   R. Jain, Z. Ji, S. Upadhyay, J. Watrous, \emph{QIP = PSPACE}, arXiv:0907.4737.

\bibitem{MariusHabi} M. Junge, \emph{Factorization theory for Spaces of Operators},
Habilitationsschrift Kiel, 1996; see also: Preprint server of the
university of southern Denmark 1999, IMADA preprint: PP-1999-02.


\bibitem{JungeParcet1} M. Junge, J. Parcet, \emph{Rosenthal's theorem for subspaces of noncommutative Lp}, Duke Math. J. 141, 75-122 (2008).

\bibitem{JungeParcet2} M. Junge, J. Parcet, \emph{Mixed-norm inequalities and operator space Lp embedding theory}, To appear in Mem. Amer. Math. Soc.

\bibitem{JungeParcet3} M. Junge and J. Parcet, \emph{A transference method in quantum probability}. Preprint.


\bibitem{JPX} M. Junge, J. Parcet and Q. Xu, \emph{Rosenthal type inequalities for free chaos}, Ann. Probab. 35 (2007), 1374-1437.


\bibitem{KRT} J. Kempe, O. Regev, B. Toner, \emph{The Unique Games Conjecture with Entangled Provers is False
}, Proceedings of 49th Annual IEEE Symposium on Foundations of
Computer Science (FOCS 2008), quant-ph/0710.0655 (2007).


\bibitem{KKMTV} J. Kempe, H. Kobayashi, K. Matsumoto, B. Toner and T. Vidick, \emph{Entangled games are
hard to approximate}, arXiv:0704.2903v2 (2007).


\bibitem{KV} S. Khot and N. K. Vishnoi. \emph{The unique games conjecture, integrality gap for cut problems
and embeddability of negative type metrics into $\ell_1$}. In Proc. 46th IEEE Symp. on
Foundations of Computer Science, pages 53-62. (2005).

\bibitem{Kraus}  B. Kraus, N. Gisin, R. Renner, \emph{Lower and upper bounds on the secret key rate for QKD protocols using one--way classical communication}, Phys. Rev. Lett. 95, 080501 (2005).



\bibitem{LedouxTalagrand} M.~Ledoux, M.~Talagrand, \emph{Probability in {B}anach {S}paces},
  Springer-Verlag, 1991.

\bibitem{MarcusPisier} M.B. Marcus, G. Pisier, \emph{Random Fourier series
with applications to Armonic Analysis}, Annals of Math. Studies,
\textbf{101}, Princeton Univ. Press, (1981).

\bibitem{Mas2}  Ll. Masanes, R. Renner, A. Winter, J. Barrett, M.
Christandl, \emph{Security of key distribution from causality
constraints}, quant-ph/0606049 (2006).


\bibitem{Mas} L. Masanes, \emph{Universally-composable privacy amplification from causality constraints}, Phys. Rev. Lett. 102, 140501 (2009).


\bibitem{Massar} S. Massar. Nonlocality, closing the detection loophole, and communication complexity. Physical
Review A, 65:032121, 2002.

\bibitem{Massar2} S. Massar, S. Pironio, \emph{Violation of local realism vs detection
efficiency}, Phys. Rev. A 68, 062109 (2003).

\bibitem{NPA1} M. Navascués, S. Pironio, A. Acín, \emph{Bounding the set of quantum correlations},
Phys. Rev. Lett. 98, 010401 (2007)


\bibitem{NPA2} M. Navascués, S. Pironio, A. Acín, \emph{A convergent hierarchy of semidefine programs characterizing the set of quantum correlations}, New J. Phys. 10, 073013 (2008).

\bibitem{Paulsen} V. I. Paulsen, \emph{Completely Bounded Maps and Operator Algebras}, Cambridge Studies
in Advanced Mathematics 78, Cambridge University Press, Cambridge, 2003.

\bibitem{Pearle} P. M. Pearle, \emph{Hidden-variable example based upon data rejection}, Phys. Rev. D, 2:1418,
1970.


\bibitem{PWJPV} D. P\'{e}rez-Garc\'{\i}a, M.M. Wolf, C. Palazuelos, I. Villanueva and M. Junge, \emph{Unbounded violation of tripartite Bell inequalities}, Commun. Math. Phys. 279 (2), 455-486 (2008).

\bibitem{Pi} S. Pironio, \emph{Violations of Bell inequalities as lower bounds on the communication cost of nonlocal
correlations}. Physical Review A, 68(6):062102, 2003.


\bibitem{Pisierbook} G. Pisier, \emph{An Introduction to Operator Spaces},
London Math. Soc. Lecture Notes Series 294, Cambridge University
Press, Cambridge 2003.

\bibitem{Pisierbook2} G. Pisier, \emph{Non-Commutative Vector Valued Lp-Spaces and Completely p-Summing Maps}, Asterisque, 247 (1998).

\bibitem{Pisier3} G. Pisier, \emph{Factorization of linear operators and geometry of Banach spaces}, CBMS 60 (1986).

\bibitem{Rao}, A. Rao, \emph{Parallel repetition in projection games and a concentration bound}, STOC2008.
\bibitem{Raz} R. Raz, \emph{A Parallel Repetition Theorem}, SIAM Journal on Computing 27, 763-803 (1998).



\bibitem{Ruan} Z-J. Ruan, J. Funct. Anal. \textbf{76} 217 (1988).

\bibitem{ShZh} Y. Shi and Y. Zhu. \emph{Tensor norms and the classical communication complexity of bipartite quantum
measurements}. SIAM Journal on Computing, 2008. To appear

\bibitem{Shor} P.W. Shor, J. Preskill, \emph{Simple Proof of Security of the BB84 Quantum
Key Distribution Protocol}, Phys. Rev. Lett. 85, 441-444, (2000).

\bibitem{Tomczak} N. Tomczak-Jaegermann, {\it Banach-Mazur Distances
and Finite Dimensional Operator Ideals}, Pitman Monographs and
Surveys in Pure and Applied Mathematics 38, Longman Scientific and
Technical, 1989.

\bibitem{Tsirelson} B.S. Tsirelson, Hadronic Journal Supplement 8:4,
329-345 (1993).


\bibitem{Vertesi}  T. Vertesi, K.F. Pal, \emph{Bounding the dimension of bipartite quantum
systems}, arXiv:0812.1572

\bibitem{Wehner} S. Wehner, M. Christandl, A. C. Doherty, \emph{ A lower bound on the dimension of a quantum system given measured
data}, Phys. Rev. A 78, 062112 (2008).


\bibitem{WernerWolf}  R.F. Werner, M.M. Wolf, \emph{Bell inequalities and Entanglement},
Quant. Inf. Comp., \textbf{1} no. 3, 1-25 (2001).


\bibitem{Wolf} M.M. Wolf, D. P\'{e}rez-Garc\'{\i}a, \emph{Assessing dimensions from
evolution},  Phys. Rev. Lett. 102, 190504 (2009).




\end{thebibliography}
\end{document}